\documentclass[12pt,draftcls,onecolumn]{IEEEtran}
\usepackage{latexsym, amsthm, amssymb, amsmath,part}
\usepackage{graphicx}

\title{Stabilizability and Norm-Optimal Control Design subject to Sparsity Constraints}

\author{ \c{S}erban Sab\u{a}u and\thanks{\c{S}. Sab\u{a}u
is with the Electrical and Computer Engineering Dept.,
Univ. of Maryland  College Park, MD 20742-3285, USA
email: sserban@isr.umd.edu.}
{Nuno C. Martins \thanks{N. C. Martins is with the Faculty of the Electrical and Computer Engineering Dept. and  the Institute for Systems Research, University of Maryland, College Park, MD 20742-3285, USA.
email: nmartins@isr.umd.edu. This work is partially funded by NSF CPS grant CNS  0931878, AFOSR grant FA955009108538, ONR UMD-AppEl Center, and the Multiscale Systems Center, one of six research centers funded under the Focus Center Research Program. }}}

\renewcommand{\tilde}{\widetilde}

\newcommand{\FF}{{{\rm I \kern -0.2em R}}}
\newcommand{\RR}{{{\rm I \kern -0.2em R}}}

\newcommand{\hlinee}{\noindent \begin{tabular}{p{\textwidth}} \hline\ \\\end{tabular}}

\newcommand{\CC}{{{\mbox{\rm \hspace*{0.05ex}
\rule[.18ex]{.18ex}{1.24ex} \kern -.65em C}}}} 
\newcommand{\rank}{\operatorname{rank}}

\newcommand{\bea}{\begin{eqnarray}}
\newcommand{\eea}{\end{eqnarray}}

\newcommand{\diag}{\mathrm{diag}}

\newtheorem{theorem}{Theorem}[section]
\newtheorem{rem}[theorem]{Remark}
\newtheorem{prop}[theorem]{Proposition}
\newtheorem{defn}[theorem]{Definition}
\newtheorem{coro}[theorem]{Corollary}
\newtheorem{problem}[theorem]{Problem}
\newtheorem{assumption}[theorem]{Assumption}

\newcommand{\ba}{\left[ \begin{array}}
\newcommand{\baa}{\begin{array}}
\newcommand{\ea}{\end{array} \right]}
\newcommand{\eaa}{\end{array}}

\newcommand{\be}{\begin{equation}}
\newcommand{\ee}{\end{equation}}
\newcommand{\bb}{\begin{equation}\label}

\newcommand{\boC}{\mathbb{C}}

\newcommand{\la}{\lambda}
\newcommand{\La}{\Lambda}
\renewcommand{\tt}{\times}

\def\math#1{\ifmmode{#1} \else {$#1$}\fi}

\newcommand{\sg}{\ifmmode \Sigma \else $\Sigma$ \fi}

\date{}

\begin{document}
\maketitle

\begin{abstract}

Consider that a linear time-invariant (LTI) plant is given and that we wish to design a stabilizing controller for it.
 Admissible controllers are LTI  and must comply with a pre-selected sparsity pattern.
The sparsity pattern is assumed to be quadratically invariant (QI) with respect to the plant, which, from prior results, guarantees that there is a
convex parametrization of all admissible stabilizing controllers provided that an initial admissible stable stabilizing controller is provided. This paper addresses the previously unsolved problem
of determining necessary and sufficient conditions for the existence of an admissible stabilizing controller. 
The main idea is to cast the existence of such a controller as the feasibility of an exact model-matching problem with stability restrictions, which can be
tackled using existing methods. Furthermore, we show that, when it exists, the solution of the model-matching problem can be
used to compute an admissible stabilizing controller. This method also leads to a convex parametrization
that may be viewed as an extension of Youla's classical approach so as to incorporate sparsity constraints. 
Applications of this parametrization on the design of norm-optimal controllers via convex methods are also explored.   An illustrative example is provided, and a special case is discussed for which the exact model matching problem has a unique and 
easily computable solution.
\end{abstract}

\section{Introduction}

In this paper, we deal with the problem of  output--feedback stabilization for  linear  time--invariant (LTI) plants using sparsity-constrained LTI controllers\footnote{For an interpretation of sparsity constraints in terms
of the interconnection structure of distributed controller, see \cite[Section III B]{Rotkowitz:2012eq}.}. The sparsity constraints are specified by a binary matrix with the same number of rows and columns as the controller. More specifically, entries of the controller must be zero whenever the corresponding element of the constraint matrix is zero, and are unrestricted otherwise.

\subsection{Previous Results} 

The convex parametrization \cite{Youla:1976fs} proposed by Youla,  which spans  {\em all} LTI controllers that stabilize a prescribed LTI plant, popularized the so--called factorization approach \cite{V} to the analysis and synthesis of LTI feedback systems. The methods proposed in \cite{DESOER:1980vk} cast the search space in a ring, which provides additional insight and tools rooted on algebraic methods.  However useful in expressing
the design of norm-optimal controllers as convex programs, Youla's parametrization does not allow for sparsity constraints on the controller.  The recent work in \cite{Voulgaris:2000ek, Voulgaris:2001bt, Qi:2004kz, Bamieh:2005by} partially bridges this gap by identifying properties that the sparsity pattern of the plant and the one imposed on the controller must satisfy so that a convex parametrization of all stabilizing controllers may exist. 
These recently discovered methods spring from invariance principles that are valid in the presence of what the authors define as {\em funnel} causality, and their validity extended to the more
general class \cite{Rotkowitz:2010ta} of \emph{quadratically invariant} sparsity patterns \cite{Rotkowitz:2006kz, Rotkowitz:2006gi}. The invariance condition in \cite{Rotkowitz:2006kz, Rotkowitz:2006gi} can be readily checked via an algebraic test, which, if true, assures that if there exists a stable stabilizing controller that satisfies the sparsity constraint then the set of all sparsity-constrained stabilizing controllers admits a convex parametrization based on a modification of the one in \cite{ZAMES:1981wb}.  Subsequent work \cite{Sabau:2009fi} has provided another convex parametrization that is guaranteed to exist under quadratic invariance, provided that a stabilizing controller that satisfies the sparsity constraint is given, and unlike prior work is not required to be itself stable. It has also been shown recently \cite{Lessard:2011ur} that quadratic invariance of the set of controllers is necessary for the existence of the convex parametrization proposed in \cite{Rotkowitz:2006kz, Rotkowitz:2006gi}.

\subsection{Contributions of this paper:} 
The main results of this paper are motivated by the following problem. 

\begin{problem} \label{central_problem} Consider that an LTI plant and a commensurate quadratically invariant sparsity constraint are given. Is the plant stabilizable
by an LTI controller that satisfies the sparsity constraint? If one exists then compute it and give a convex parametrization of all stabilizing sparsity-constrained controllers.
\end{problem}

For a given plant, in this paper we establish necessary and sufficient conditions for the existence of a stabilizing LTI controller, subject to pre-specified
quadratically invariant sparsity constraints. If one exists then our analysis also provides a method to construct a stabilizing controller that respects the sparsity constraints.
Since all existing convex parametrizations presuppose prior knowledge of a stabilizing sparsity constrained controller \cite{Rotkowitz:2006kz}, our results bridge an
important gap in the design process.
 
In our solution method, the necessary and sufficient conditions mentioned in Problem~\ref{central_problem}  are cast as the {\em existence} of a certain doubly coprime factorization \cite{NETT:1984un, Khargonekar:1982vm} of the plant that has additional constraints on the factors. We  show that determining when such a factorization exists, and if so computing one, is equivalent to solving an exact model--matching problem with stability restrictions \cite{Wolo}. We also give a convex {\em Youla-like} parametrization of  the set of all sparsity constrained stabilizing controllers by imposing additional constraints on the Q-parameter that require that it satisfies a certain homogeneous system of linear equations over the field of transfer functions. Unlike prior parametrizations that require an initial stable stabilizing controller that satisfies the sparsity constraint, our Youla-like parametrization does not require an initial controller and it is valid even when the plant is non-strongly stabilizable.

\subsection{Paper organization:} Including the introduction, this paper has six sections. Section \ref{Sec:2} states definitions and preliminary results used throughout the paper, while Section \ref{Sec:3} reviews the notation and state of the art on design of sparsity constrained controllers. The main results of this paper are in Section \ref{Sec:4}, where we formulate the necessary and sufficient conditions for stabilizability as the existence of solutions to an exact model matching problem \cite{Wolo}. We also propose methods to compute a sparsity-constrained stabilizing controller, when one exists, along with a numerical example. In addition, we present an associated convex parametrization of all stabilizing sparsity-constrained controllers that is obtained by imposing subspace constraints on Youla's parameter.  These results are specialized in Section~\ref{Sec:5} to plants that admit a structured doubly coprime factorization that we  denominate  {\em Input/Output Decoupled}.  We show that this special factorization may simplify the application of our results and provide additional insights. The paper ends with conclusions in Section \ref{Sec:6}.

{\bf Comparison with prior publications by the authors:}  Some of the results presented here have been published in preliminary form in \cite{Sabau:2011dy} and \cite{WrongAllerton10}. In particular, parts of Sections \ref{Sec:3} and~\ref{Sec:4} have been discussed with less detail  in \cite{Sabau:2011dy}. The discussion in \cite{Sabau:2011dy}  assumes block partitioning of the matrices, while, in this paper, partitioning is assumed only in Section \ref{Sec:5}. In contrast with \cite{Sabau:2011dy}, where we provide two simple examples, in Section \ref{ANE} of this paper we provide a more involved example on how to construct a sparsity-constrained controller. An abridged version of Section~\ref{Sec:5} was discussed in \cite{WrongAllerton10} in which there was a technical flaw. More specifically, we later found that Lemma 3.1 of \cite{WrongAllerton10} is incorrect and a correct and detailed discussion is provided in Section~\ref{Sec:5} and Appendix-I of this paper.  This paper also establishes a strong connection between the approaches of Sections~\ref{Sec:4} and~\ref{Sec:5}.

\section{Preliminaries}
\label{Sec:2}

We focus on the standard feedback configuration of Fig. \ref{2Block}, where $G$ is an LTI plant and $K$ is an LTI controller that are finite dimensional and operate in either continuous or discrete--time. Here, $\nu_1$ and $\nu_2$ are the input disturbance and sensor noise, respectively. In addition, $u$  and $y$ are the control and measurable output vectors, respectively.

\begin{figure}
\begin{centering}
\includegraphics{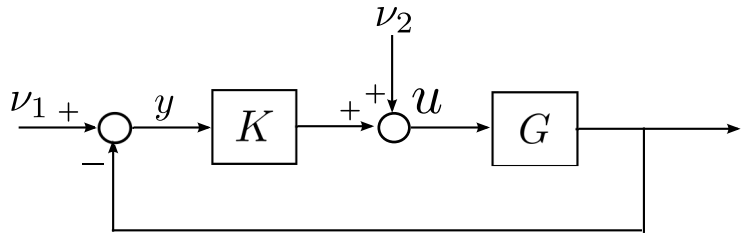}
\caption{Standard unity feedback configuration}
\label{2Block} 
\end{centering}
\end{figure}

\newpage
We adopt the following notation:

\begin{tabular}{l l}
& \\
 $\mathbb{R}(\la) $ & Set of all real--rational transfer functions.  \\
 $\mathbb{R}(\la)^{n \times q} $ & Set of $n \times q$ matrices with entries in $\mathbb{R} (\la)$ .\\
 TFM & Transfer function matrix, or, equivalently, $\bigcup_{n,q} \mathbb{R}(\la)^{n \times q}$.\\
 $\Omega$ & Stability region for TFMs.\\
 $\mathbb{A} (\la)$ & Subset of $\mathbb{R}(\la) $ whose entries have poles in $\Omega$.\\
 $\mathbb{A}(\la)^{n \times q} $ & Set of $n \times q$ stable TFMs.\\
 $\bar{\mathbb{A}}(\lambda)$ & Set of stable TFMs, or, equivalently, $\bigcup_{n,q} \mathbb{A}(\la)^{n \times q}$.\\
 $\mathbb{B}$ & The set $\{0,1\}$.\\
 & 
 \end{tabular}

We also adopt the following assumptions and conventions:

\begin{tabular}{l l} & \\
$m$ & Dimension of $y$. \\
$p$ & Dimension of $u$.\\
$G$ & The plant is a TFM with \underline{strictly proper} entries. \\
$K$ & The LTI controller is an element of $\mathbb{R}(\la)^{p \times m}$. \\
$H(G,K)$ & The TFM from $\displaystyle [\nu_1^T \; \; \nu_2^T\; \: ]^T$ to $\displaystyle [y^T \;\; u^T \;\: ]^T$. \\
$\otimes$ & Kronecker product.\\
$\overline{1,q}$ & $\{1, 2,\cdots, q\}$ \\
&
\end{tabular}

The indeterminate $\la$ is either $s$ for continuous--time or $z$ for discrete--time systems, respectively. The $\lambda$ argument of a TFM is \underline{often omitted} when its presence is clear from the context.  
 If the transfer matrix $H(G,K)$ is {\em stable} we say that $K$ is a {\em stabilizing controller} of $G$, or equivalently  that $K$ {\em stabilizes} $G$. If a stabilizing controller of $G$ exists, we say that $G$ is {\em stabilizable}.


\subsection{Coprime and Doubly Coprime Factorizations for LTI Systems} \label{2}

A {\em right coprime factorization} (RCF) of $G$  over $\Omega$ is a fractional representation of the form $G = NM^{-1}$, with $N \in \mathbb{A}^{m \times p} $ and $M \in \mathbb{A}^{p \times p} $, and for which there exist $X \in \mathbb{A}^{p \times m} $ and $Y \in \mathbb{A}^{p \times p} $ satisfying ${Y}M + {X}N = I$  (\cite[Ch.~4, Corollary~17]{V}). Analogously, a {\em left  coprime factorization} (LCF) of $G$ (over $\Omega$) is defined by $G = \tilde{M}^{-1}\tilde{N}$,  with $\tilde{N} \in \mathbb{A}^{m \times p} $ and $\tilde{M} \in \mathbb{A}^{m \times m} $, satisfying $\tilde{M}\tilde{Y} + \tilde{N}\tilde{X} =I$ for  $\tilde{X} \in \mathbb{A}^{p \times m} $ and $\tilde{Y} \in \mathbb{A}^{m \times m} $.  Due to the natural interpretation of the coprime factorizations  as fractional representations, the invertible $\tilde M$ and $M$ factors are sometimes called the ``denominator'' TFMs of the coprime factorization.

\begin{defn} 
A collection of eight  stable TFMs $\big(M, N$, $\tilde M, \tilde N$, $X, Y$, $\tilde X, \tilde Y\big)$
is called a  {\em doubly coprime factorization} (DCF) of $G$ over $\Omega$  if   $\tilde M$ and $M$ are invertible, yield the 
following factorizations:
$$G=\tilde{M}^{-1}\tilde{N}=NM^{-1}$$  and satisfy the following equality  (B\'{e}zout's identity):
\bb{dcrel}
\ba{cc}  Y &  X \\ -\tilde N &  \tilde M\ea
\ba{cc} M & - \tilde X \\ N &  \tilde Y \ea = I_{m+p}.
\ee
\end{defn}
To avoid excessive terminology, we refer to doubly coprime factorizations over $\Omega$ simply as doubly coprime factorizations {\bf (DCFs)} \cite[Ch.4, Remark pp. 79]{V}.

\begin{theorem}
\label{Youlaaa}
{\bf (Youla)} \cite[Ch.5, Theorem 1]{V} Let  $\big(M, N$, $\tilde M, \tilde N$, $X, Y$, $\tilde X, \tilde Y\big)$ be a DCF of $G$. Any stabilizing controller can be written as: 
\begin{equation}
\label{YoulaEq}
K_Q =\tilde{X}_Q \tilde {Y}_Q^{-1} 
   = Y_Q^{-1} X_Q \\
\end{equation}
for some $Q$ in $\mathbb{A}^{p\times m}$, where
$X_Q$, $\tilde{X}_Q$, $Y_Q$ and $\tilde{Y}_Q$ are defined as:
\begin{eqnarray}
X_Q & \overset{def}{=} & X+Q \tilde{M} \\
\tilde{X}_Q & \overset{def}{=} & \tilde{X}+MQ \\
Y_Q & \overset{def}{=} & Y - Q \tilde{N} \\
\tilde{Y}_Q & \overset{def}{=} & \tilde{Y}-NQ
\end{eqnarray} It also holds that $K_Q$ stabilizes $G$ for any $Q$ in $\mathbb{A}^{p\times m}$.
\end{theorem}

\begin{rem} 
\label{rem:232pm27feb2012}
The following identity shows that $\big(M, N, \tilde M, \tilde N, X_Q, Y_Q, \tilde X_Q, \tilde Y_Q \big)$ in Theorem~\ref{Youlaaa} is also a DCF of $G$:
\begin{equation}\label{dcrelQ}
\ba{cc}  (Y  - Q \tilde N) &  (X + Q \tilde M) \\ -\tilde N &  \tilde M \ea
\ba{cc} M & -(\tilde X +MQ) \\ N &  (\tilde Y -NQ) \ea = I_{m+p}, \qquad Q \in \mathbb{A}^{p \times m} 
\end{equation}
\end{rem}

\section{Feedback  Control subject to  Sparsity Constraints}
\label{Sec:3}

The precise formulation of the sparsity constrained stabilization problem  is  achieved  by imposing  a certain pre--selected sparsity pattern on the set of admissible  stabilizing controllers. 

\subsection {Specifications of Sparsity Constraints on LTI Controllers} \label{Spars}

For the boolean algebra, the operations $(+, \cdot )$ are defined as usual: $0+0=0\cdot 1=1 \cdot 0= 0 \cdot 0=0$ and $1+0=0+1=1+1=1\cdot 1=1$. By a binary matrix we mean a matrix whose entries belong to the set $\mathbb{B}$.  With the usual extension of notation, $\mathbb{B} ^{q \times l}$ stands for the set of all binary matrices with $q$ rows and $l$ columns. The addition and multiplication of binary matrices is carried out in the usual way, keeping in mind that
the binary operations  $(+, \cdot )$ follow the boolean algebra.  Binary matrices are marked with a ``$\mathrm{bin}$'' superscript, in order to distinguish them from transfer function matrices over $ \mathbb{R}(\la)$. Furthermore, for binary matrices of the same dimension, the notation $\displaystyle A^\mathrm{{bin}} \leq B^\mathrm{{bin}}$ means that $a_{ij} \leq b_{ij}$ holds entrywise for all $i$ and $j$. 

A binary matrix may be associated with a TFM of the same dimension, whereby each entry of the binary matrix corresponds to an entry of the TFM. The following definitions introduce operators that will be used to
establish a correspondence between binary matrices and the sparsity pattern of $G$ or sparsity constraints imposed on $K$. 

\begin{defn} ({\bf Pattern operator})
Given $A$ in $\mathbb{R}(\lambda)^{q \times l}$, we define $ \wp(A)  \in \mathbb{B} ^{q \times l}$  as follows:
\begin{equation}
\label{Pattern}
\wp(A)_{ij}\overset{def}{=}\left \{
\begin{array}{l} 0  \quad \mathrm{if } \qquad  A_{ij}=0  \\
                     1 \quad \mathrm{otherwise} \end{array} \right.  \qquad   i,j \in \overline{1,q}  \times \overline{1,l}
\end{equation}
\end{defn}

\begin{defn} ({\bf Sparse operator})
Conversely, for any binary matrix $\displaystyle A^\mathrm{{bin}}$  in $\mathbb{B} ^{q \times l}$, we define the following linear subspace:
\begin{equation}
\label{Sparse}
 \mathbb{S}(A^\mathrm{{bin}}) \overset{def}{=} \Big{\{} A \in  \mathbb{R}(\la)^{q\times l} \big{|} \; \wp (A) \leq A^\mathrm{{bin}} \Big{\}}
 \end{equation}
\end{defn}


\begin{defn} Given $K^{\mathrm{bin}}$ in $\mathbb{B}^{p \times m}$, the {\em sparsity constraint} $\mathcal{S}$ is defined as \cite{Rotkowitz:2006kz}:
\begin{equation}
\label{DefnS}
\mathcal{S} \overset{def}{=}\mathbb{S}(K^{\mathrm{bin}}),
\end{equation}
\end{defn}

\noindent  Hence,  $\mathcal{S}$ is the subspace of all controllers $\displaystyle K$ in $\displaystyle  \mathbb{R}(\la)^{p\times m}$ for which $K_{ij}=0$ whenever $\displaystyle K^\mathrm{{bin}}_{ij}=0$. 

We assume that $\mathbb{S}$ and $\wp$ act on $G$ and $G^{\mathrm{bin}}$ on an analogous way as above, leading to the following definitions.

\begin{defn}
The following is the sparsity pattern of $G$ ($G^{\mathrm{bin}} \in \mathbb{B}^{m \times p}$):
\begin{equation}
\label{DefnPbin}
    G^{\: \mathrm{bin}} \overset{def}{=}\:\wp(G)
\end{equation}
\end{defn}

\begin{rem}
From matrix multiplication (in the boolean algebra), we conclude that the following holds: 
\begin{equation}
\label{ineg}
\wp(K\:G) \leq \wp(K)\:\wp(G)
\end{equation}
\end{rem}


\subsection{Quadratic Invariance}
\label{QIsubsection}

 

\begin{defn}
\label{QI} \cite[Definition 2]{Rotkowitz:2006kz}
The sparsity constraint $\mathcal{S}$ is called {\em quadratically invariant} {\bf (QI)} under the plant $G$ if
\begin{equation} \label{EqQI}
 K  G K \in \mathcal{S}, \quad  K\:  \in \mathcal{S}.
\end{equation}
\end{defn}

\begin{rem} The following conditions are equivalent to (\ref{EqQI}) \cite{Rotkowitz:2006kz}:
\begin{itemize}
\item $\wp(KGK) \leq K^{\mathrm{bin}}, \text{ } K \in \mathcal{S}$
\item $K^{\mathrm{bin}}G^{\mathrm{bin}}K^{\mathrm{bin}} \leq K^{\mathrm{bin}}$
\end{itemize}
\end{rem}

\begin{defn} \label{fdbckptrans}
Define the feedback transformation $h_G: \mathbb{R}(\lambda)^{p \times m} \rightarrow \mathbb{R}(\lambda)^{p \times m}$ of $G$ with $K$, as follows:
\begin{equation}
\label{hG}
h_G(K) \overset{def}{=}K \big{(}I+ G K \big{)}^{-1}, \qquad K \in  \mathbb{R}(\lambda)^{p \times m}
\end{equation}
\end{defn}

\begin{rem} \label{invfdbckptrans}
The feedback transformation $h_G(\cdot) $ is invertible, and its inverse is given by
\begin{equation}
\label{invhG}
h^{-1}_G(K) \overset{def}{=}K \big{(}I- G K \big{)}^{-1}
\end{equation}
\end{rem}
\begin{proof} Note that $h_G(\cdot) $ is well--posed because $K$ is proper and $G$ is strictly proper. The rest of the proof follows by direct algebraic computations and is omitted for brevity.
\end{proof}

The following result is used throughout the paper. 
\begin{theorem} \cite[Theorem~14]{Rotkowitz:2006kz}
\label{Mike's}
Given a sparsity constraint $\mathcal{S}$, the following equivalence holds:
\begin{equation} \label{Rotk}
\mathcal{S}\; \text{is QI under} \; G \Longleftrightarrow h_{G}(\mathcal{S} ) = \mathcal{S}
\end{equation} where we adopt the following abuse of notation: $$ h_{G}(\mathcal{S} )\overset{def}{=} \{h_G(K) | K \in \mathcal{S} \}$$
\end{theorem}
An alternative algebraic proof of Theorem (\ref{Rotk})  is given in \cite[Theorem~9]{Lessard:2010be}.

\begin{rem} \label{QIminG} The set $\mathcal{S}$ is QI under the given plant $G$ if and only if $\mathcal{S}$ is QI under $-G$. This implies, via (\ref{Rotk}) above, that $\mathcal{S}$ is QI under  $G$ if and only if $\displaystyle h^{-1}_{G}(\mathcal{S} ) = \mathcal{S}$, where $h^{-1}_{G}(\mathcal{S} ) \overset{def}{=} \{h^{-1}_G(K) | K \in \mathcal{S} \}$.
\end{rem}



\section{Main Result}
\label{Sec:4}

Given a QI sparsity constraint $\mathcal{S}$, in Theorem~\ref{truemain}  we develop necessary and sufficient  conditions for the existence of a {\em stabilizing}  controller in $\mathcal{S}$. These conditions are formulated in terms of the existence of a  doubly coprime factorization of the plant in which the factors satisfy additional constraints.  Such a  factorization (when it exists) is equivalent to solving an exact model matching problem with stability \cite{Wolo} restrictions, which has been previously investigated in the control literature.  The following preparatory result will be used throughout this Section.


\begin{prop}
\label{Sebe}
Let $\big(M, N,\tilde M, \tilde N,X, Y, \tilde X, \tilde Y\big)$ be a given DCF of $G$. The following identities hold:
\begin{equation}
\label{SebeEq}
MX_Q=K_Q\big(I+GK_Q \big)^{-1} , \quad  \tilde X_Q \tilde M=K_Q\big(I+GK_Q \big)^{-1}, \qquad Q \in \mathbb{A}^{p\times m}
\end{equation} 
\end{prop}
\begin{proof}   We  proceed to verifying that $MX_Q=K_Q\big(I+GK_Q \big)^{-1}$ is true, while the proof that $X_Q \tilde M=K_Q\big(I+GK_Q \big)^{-1}$ holds is omitted because it is analogous.
From $K_Q=Y_Q^{-1}X_Q$ and $G=NM^{-1}$, we get that $K_Q\big(I+GK_Q \big)^{-1} = \big(I+Y_Q^{-1}X_QN M^{-1} \big)^{-1}Y_Q^{-1}X_Q$, where we used the fact
that $K_Q\big(I+GK_Q \big)^{-1} = \big(I+K_QG \big)^{-1}K_Q$. Finally, using B\'{e}zout's identity we find that $\big(I+Y_Q^{-1}X_QN M^{-1} \big)^{-1} =\big(I+Y_Q^{-1}(I-Y_QM) M^{-1} \big)^{-1}=MY_Q$, which by direct substitution in $\big(I+Y_Q^{-1}X_QN M^{-1} \big)^{-1}Y_Q^{-1}X_Q $ concludes the proof. 
\end{proof}

\newpage
The following Theorem is a main result of this paper. \\
\hlinee  
\begin{theorem}
\label{truemain} Let  $\big(M, N$, $\tilde M, \tilde N$, $X, Y$, $\tilde X, \tilde Y\big)$ be a DCF of $G$ and $\mathcal{S}$ be a 
 QI sparsity constraint.  
 \begin{itemize}
 \item {\bf Sufficiency:}
If $Q$ in $\mathbb{A}^{p\times m}$ is such that at least one of the following inequalities holds\footnote{In fact, it also follows from the statement of the Theorem that either both inequalities hold or none.}:
\begin{subequations}
\label{adouaEq2}
\begin{align}
\wp (\tilde X_Q \tilde M)  \leq  K^{\mathrm{bin}} \\ 
 \wp(M X_Q)  \leq  K^{\mathrm{bin}}
\end{align}
\end{subequations} then $K_Q$ is a stabilizing controller in $\mathcal{S}$. 

\item {\bf Necessity:} If there is a stabilizing controller in $\mathcal{S}$ then there exists some $Q$ in $\mathbb{A}^{p\times m}$
for which both inequalities in  (\ref{adouaEq2}) hold and, in addition, the controller can be written as $K_Q$.
\end{itemize}
\end{theorem}  \hlinee
\begin{proof} {\bf Necessity:}  Suppose that there exists a stabilizing controller in  $\mathcal{S}$, then, as a consequence of Youla's Theorem~\ref{Youlaaa}, such a controller can be written as $K_Q$ for some $Q$ in $\mathbb{A}^{p\times m}$. We now use the fact that $K_Q$ is in $\mathcal{S}$ to prove that both inequalities in  (\ref{adouaEq2}) hold.  According to Proposition~\ref{Sebe} we get from (\ref{SebeEq}) that 
\begin{equation} \label{ec8}
\tilde X_Q \tilde M=K_Q\big(I+GK_Q \big)^{-1}
\end{equation}
\noindent We apply the $\wp$ operator (\ref{Pattern}) on both sides of equation (\ref{ec8})  and using  Definition~\ref{fdbckptrans} we find that 
$\wp( \tilde X_Q \tilde M)=\wp(h_G(K_Q))$. Since $\mathcal{S}$ is QI and $K_Q$ is in $\mathcal{S}$, it follows from (\ref{Rotk}) that $h_G(K_Q)$ belongs to $\mathcal{S}$ and $\wp(h_G(K_Q))\leq K^{\mathrm{bin}}$, which leads to $\wp (\tilde X_Q \tilde M) \leq K^{\mathrm{bin}}$. Similarly, we employ (\ref{SebeEq}) to get that $\wp(MX_Q)=\wp(h_G(K_Q))$ in order to finally obtain that  $\wp(M X_Q)\leq K^{\mathrm{bin}}$.

{\bf Sufficiency:} Take each side of (\ref{ec8}) as an argument for $h_G^{-1}(\cdot)$ in order to get via Definition~\ref{fdbckptrans}  that  $h_G^{-1}(\tilde X_Q \tilde M)=h_G^{-1} (h_G(K_Q))$ and equivalently that $K_Q=h_G^{-1}(\tilde X_Q \tilde M)$. In addition, it follows from Remark~\ref{invfdbckptrans} and Remark~\ref{QIminG} that $h_G^{-1}(\mathcal{S})=\mathcal{S}$, which, from the assumption that $\wp( \tilde X_Q \tilde M)\leq K^{\mathrm{bin}}$,  implies that $K_Q=h_G^{-1}( \tilde X_Q \tilde M)$ is in $\mathcal{S}$ .  The fact that $K_Q$ is stabilizing follows from Youla's Theorem~\ref{Youlaaa}.
The sufficiency with respect to $\wp(MX_Q)\leq K^{\mathrm{bin}}$ follows by a similar line of proof and so is omitted for brevity.
\end{proof}

\newpage
\subsection{Controller Synthesis as An Exact Model--Matching Problem  with Stability Restrictions} \label{modmatching}

Henceforth, given a matrix $V$ with $n$ rows and $q$ columns, we adopt the following notation:

\begin{tabular}{l l}
& \\
$\mathrm{vec}(V)$ & $\vec{v} = \mathrm{vec}(V)$ gives $\vec{v}_{(i+(j-1)n)}=V_{i,j} $\\
 $\mathrm{diag}(V)$ & $\Delta = \mathrm{diag}(V)$ is diagonal and $\Delta_{ii} =\vec{v}_i$  \\
 & 
 \end{tabular}
 
In this section, we will outline a method (based on Theorem~\ref{truemain} above) for the computation of a stabilizing controller subject to a pre-selected QI sparsity constraint $\mathcal{S}$ (whenever such a controller exists). Given a DCF of $G$, which can be computed using the standard state--space techniques in \cite{NETT:1984un,Khargonekar:1982vm}, our goal is to obtain $Q$ in $\mathbb{A}^{p\times m}$ such that (\ref{adouaEq2}) is satisfied.

Our approach is based on the realization that (\ref{adouaEq2}) can be cast 
as the feasibility of an exact model-matching problem \cite{Wolo} with respect to $Q$ in $\mathbb{A}^{p\times m}$.  
This correspondence is stated precisely in the following Theorem, while Section \ref{sec:349pm1March2012}
provides more details and references on the computation and tests for the existence of solutions of the exact model matching
problem.


\hlinee
\begin{theorem} \label{mm} Consider that a DCF of $G$ $(M,N,\tilde{M},\tilde{N},X,Y,\tilde{X},\tilde{Y})$ is given and that a QI sparsity constraint $\mathcal{S}$ is pre-selected via a choice of $K^{bin}$ in $\mathbb{B}^{p \times m}$. The existence of a stabilizing controller in $\mathcal{S}$ is equivalent to the existence of $Q$ in $\mathbb{A}^{p \times m}$ for which
at least one of the following equivalent equalities holds:
\begin{subequations} \label{mmEq}
\begin{align}
\Phi  \big(M^T \otimes \tilde M\big) \mathrm{vec}(Q)  +   \Phi \: \mathrm{vec} \big( \tilde X \tilde M \big)  =  0 \\
\Phi  \big(M^T \otimes \tilde M\big) \mathrm{vec}(Q)  +  \Phi \: \mathrm{vec} \big( M X \big)  =  0
\end{align}
\end{subequations}
\noindent where $\Phi$ is defined as:
\begin{equation} \label{defnPhi}
\Phi \overset{def}{=}I - \mathrm{diag}(K^{\mathrm{bin}})
\end{equation} In addition, if there is  $Q$ in $\mathbb{A}^{p \times m}$ that satisfies (\ref{mmEq}) then $K_Q$ is a stabilizing controller in $\mathcal{S}$.
\end{theorem}
\hlinee
\begin{proof} The proof follows by establishing an equivalence between (\ref{mmEq}) and (\ref{adouaEq2}).
We start by rewriting (\ref{adouaEq2}) as follows:
\begin{equation}
\label{eq:1055pm18April2012}
\wp \big( (\tilde X + MQ) \tilde M) =\wp \big(  M Q  \tilde M + \tilde X \tilde M ) \leq K^{\mathrm{bin}}
\end{equation}
\begin{equation} 
\label{eq:1056pm18April2012}
 \wp\big(M (X+Q\tilde M) \big) =  \wp \big( MQ \tilde M + MX) \leq K^{\mathrm{bin}}
\end{equation}
The vectorization of (\ref{eq:1055pm18April2012})-(\ref{eq:1056pm18April2012}) leads to:
\begin{equation}
\label{eq:1057pm18April2012}
\wp \bigg(  \big(M^T \otimes \tilde M \big) \mathrm{vec}(Q) +  \mathrm{vec} \big( \tilde X \tilde M \big) \bigg) \leq \mathrm{vec} \big(K^{\mathrm{bin}} \big)
\end{equation}
\begin{equation} 
\label{eq:1058pm18April2012}
 \wp\bigg( \big(M^T \otimes \tilde M\big) \mathrm{vec}(Q) +  \mathrm{vec} \big( M X \big)      \bigg) \leq  \mathrm{vec} \big( K^{\mathrm{bin}} \big)
\end{equation}

Now notice that if the $i$-th entry of $\mathrm{vec} \big(K^{\mathrm{bin}} \big)$ is zero then the corresponding entries
of the left hand side of (\ref{eq:1057pm18April2012}) and (\ref{eq:1058pm18April2012}) must both be zero, which is equivalent to (\ref{mmEq}) \end{proof}

\subsubsection{Computational considerations}
\label{sec:349pm1March2012}
Problems of the type (\ref{mmEq}) are of particular importance in linear control theory and were  formulated and proposed for the first time by Wolovich (\cite{Wolo}), who also coined the term {\em exact model--matching} in the early 1970's. Under the additional constraint that $Q$ lies in $\mathbb{A}^{p \times m}$, the problem is referred to as exact model--matching with stability restrictions (see \cite{GAO:1989uk}).  Reliable and efficient state--space algorithms for solving (\ref{mmEq}) are available in \cite{Chu:2006vq}, which also describes a method to  ascertain when a solution exists and consequently, from Theorem~\ref{mm}, decide when a stabilizing controller in $\mathcal{S}$ exists.  Given a stabilizing controller in $\mathcal{S}$ one can use the results in \cite{Sabau:2009fi,T} to obtain a convex parametrization of {\em all} stabilizing controllers in $\mathcal{S}$.  Also, since the resulting convex parametrization is affine in $Q$, one can use the tractable methods proposed in \cite{Rotkowitz:2006kz} to design norm-optimal controllers for both the disturbance attenuation and the mixed--sensitivity $\mathcal{H}_2$ problems.


\subsection{A Numerical Example} \label{ANE}

Consider the following choices for the plant $G$ and the QI sparsity constraint $\mathcal{S}$ to be imposed on the controller as specified via $K^{\mathrm{bin}}$: 
\[
G = \ba{cc}  \frac{1}{\la+4} & \frac{1}{\la-2} \\
                             \frac{1}{\la -1}  & 0  \\
                             \frac{1}{\la + 5} & \frac{1}{\la - 3}
\ea
\]

\[
K^{\mathrm{bin}} = \ba{ccc} 0 & 1 & 0 \\ 1 & 1 & 1 \\ \ea
\]

We use the state--space formulas  from \cite{NETT:1984un, Khargonekar:1982vm} to obtain the following DCF of $G$:

\[ \tilde M  = \ba{ccc} \frac{\la-2}{\la+6} & 0 & 0 \\
                                        0  & \frac{\la-1}{\la+7} & \frac{\la-3}{\la+7} \\
                                        0  &  0  & \frac{\la-3}{\la+8} \ea, \quad
    \tilde N  = \ba{cc} \frac{\la-2}{(\la+4)(\la+6)} & \frac{1}{\la+6} \\
                                           \frac{2(\la+1)}{(\la+5)(\la+7)} & \frac{1}{\la+7} \\
                                           \frac{\la-3}{(\la+8)(\la+5)} & \frac{1}{\la+8} \\ \ea ,   
                                        \]
\[
X= \ba{ccc} \frac{\la-2}{\la+6} & \frac{1}{\la+7} & \frac{\la-3}{\la+8} \\
                             \frac{1}{\la+6} & \frac{\la-1}{\la+7} & \frac{1}{\la+8}  \\ \ea ,  \]
 \[                            
Y=\ba{cc}   \frac{ \la^5 + 42 \la^4 + 617 \la^3 + 4144 \la^2 + 12968 \la + 15256}{ \la^5 + 30 \la^4 + 355 \la^3 + 2070 \la^2 + 5944 \la + 6720} &
    \frac{\la^4 + 777 \la^3 + 9557 \la^2 + 27300 \la - 13060}{\la^5 + 30 \la^4 + 355 \la^3 + 2070 \la^2 + 5944 \la + 6720} \\
  \frac{ 2 \la + 14}{\la^2 + 14 \la + 48} & \frac{\la^2 + 40 \la - 108.002}{\la^2 + 14 \la + 48} \ea,
                       \]

\[
N= \ba{cc} \frac{\la-1}{(\la+4)(\la+9)} & \frac{\la-3}{(\la+10)(\la+11)} \\
                           \frac{1}{\la+9} & 0 \\
                           \frac{\la-1}{(\la+5)(\la+9)} & \frac{\la-2}{(\la+10)(\la+11)} \\ \ea, \quad 
 M=\ba{cc} \frac{\la-1}{\la+9} & 0 \\
                             0 & \frac{(\la-2)(\la-3)}{(\la+10)(\la+11)}\ea .                          \]

\noindent The remaining factors $\tilde X$ and $\tilde Y$ of the DCF are not needed here. We now proceed to finding a
solution for (\ref{mmEq}), which, according to Theorem~\ref{mm}, leads to the conclusion that a stabilizing controller in $\mathcal{S}$ exists. In addition,
we will use the aforementioned solution to compute a stabilizing controller.

Since there are two zeros in the sparsity pattern imposed by $K^{\mathrm{bin}}$, the system of equations in (\ref{mmEq}) has two (nontrivial) equations that are satisfied by the following element of  $\mathbb{A}^{p \times m}$:           

\[
Q=\ba{ccc} 1 & 0 & 1 \\
                0 & 1 & \frac{\la+8}{\la+7}\ea
\]

The resulting stabilizing central controller $K=Y^{-1}_QX_Q$ is given by

\[
K = \ba{cc}  \frac{\la+17}{\la+7} & 0 \\
                           754\frac{(\la+5.87)(\la-0.4525)}{(\la+4)(\la+5)(\la+6)(\la+8)}  & \frac{(\la+42.5389)(\la-2.5389)}{(\la+6)(\la+8)} \ea^{-1}
                           \ba{ccc} 0 & \frac{1}{\la+7} & 0 \\
                                      \frac{1}{\la+6} & 0 & \frac{1}{\la+8} \ea,
\]

\noindent  which is in $\mathcal{S}$.


\subsection{A Youla-like Parametrization of All Sparse, Stabilizing Controllers} \label{YoulaPASC}

In this subsection, we present an alternative statement to Theorem~ \ref{mm} that clarifies the differences between it and Youla's
classical parametrization. 

\hlinee
\begin{coro} \label{coro2} Let $\mathcal{S}$ be a given QI sparsity constraint and $(M,N,\tilde{M},\tilde{N},X,Y,\tilde{X},\tilde{Y})$
a DCF of $G$. Assume that there is a stabilizing controller in $\mathcal{S}$ and let $Q_0$ in $\mathbb{A}^{p \times m}$ 
be selected to satisfy (\ref{mmEq}). Any stabilizing controller in $\mathcal{S}$ can be written as $K_Q$, where
$Q$ is obtained as:
\begin{equation}
\label{Eq:1111amapril20}
Q=Q_0+Q_{\delta}
\end{equation} for some $Q_{\delta}$ in the {\bf(convex set)} specified by the following inclusions: 
\begin{equation}
\label{Eq:1112amapril20}
 \mathrm{vec}(Q_{\delta})\in \mathrm{Null}  \Big( \Phi  \big(M^T \otimes \tilde M \big)\Big), Q_{\delta} \in \mathbb{A}^{p \times m}
\end{equation}
where $\Phi$ is the matrix  defined in (\ref{defnPhi}). 
\end{coro}
\hlinee
\begin{proof} The proof follows directly from Theorem~\ref{mm}.
\end{proof}

Corollary~\ref{coro2} unveils the fact that once any suitable $Q_0$ is found then the set of all stabilizing controllers
in $\mathcal{S}$ can be generated from the affine subspace specified by (\ref{Eq:1111amapril20})-(\ref{Eq:1112amapril20}).
Notice that in Youla's classical approach the parameter $Q$ is only required to be in $\mathbb{A}^{p \times m}$, while
the additional constraints in (\ref{Eq:1111amapril20})-(\ref{Eq:1112amapril20}) guarantee that the resulting controller
will be in $\mathcal{S}$. 

\subsection{Numerical Considerations}
 For an introduction to  linear subspaces for TFMs and vector bases of such subspaces we refer to \cite{Forney:1975vj}.  In addition, the authors of \cite{L2A} describe a systematic, state--space algorithm to determine a basis of the null space of $ \Phi  \big(M^T \otimes \tilde M \big)$. Note that the main result in \cite{L2A} enables the computation of a basis having only stable poles, by performing a column compression of the normal rank of $ \Phi  \big( M^T \otimes \tilde M \big)$ by post--multiplication with a unimodular matrix. Furthermore, this basis is also minimal, in the sense that the basis--matrix, obtained by juxtaposing the basis columns, has no Smith zeros. Hence, this may be used for the parametrization of all {\em stable} $\mathrm{vec}(Q)$ in $ \mathrm{Null}  \Big( \Phi  \big( M^T \otimes \tilde M \big)\Big)$.

For a numerical example illustrating Corollary~\ref{coro2}, from Subsection~\ref{YoulaPASC} above, we refer to Subsections~IV-C
and IV-E in \cite{Sabau:2011dy}.

\subsection{Norm-Optimal Control Design}

We now indicate how Theorem~\ref{truemain} can be used in conjunction with results from \cite{Rotkowitz:2006kz} to design norm-optimal controllers. 
In particular, given a quadratically invariant sparsity constraint $\mathcal{S}$ one may be interested in 
solving the following optimization problem:
\begin{equation}
\label{e:928p050512}
\min_{K \text{stabilizing and } K \in \mathcal{S}} \| H(G,K) \|
\end{equation} where $\| \cdot \|$ is a suitably defined operatorial norm.

Using Theorem~\ref{truemain} we can rewrite  (\ref{e:928p050512}) as follows:
\begin{equation}
\label{e:930p050512}
\min_{\begin{matrix} (MX + M Q \tilde M ) \in \mathcal{S} \\ Q \in \mathbb{A}^{p \times m} \end{matrix} } \| H(G,K_Q) \|
\end{equation} where we used the fact that the inequalities in (\ref{adouaEq2}) are equivalent to $\tilde X \tilde M + M Q \tilde M \in \mathcal{S}$ and
$MX + M Q \tilde M \in \mathcal{S}$. Notice that Theorem~\ref{truemain} guarantees that (\ref{e:928p050512}) is feasible if and only if (\ref{e:930p050512})
is feasible.

We proceed by noticing that the closed loop TFM for a given controller $K_Q$ can be written as:
\begin{subequations} \label{affine}
\begin{align}
H(G,K_Q)=\ba{cc} \tilde Y \tilde M & - \tilde Y \tilde N \\ \tilde X \tilde M & I- \tilde X \tilde N \ea + \ba{c} N \\ M \ea Q \ba{cc} \tilde M & \tilde N \ea \\
=\ba{cc} I - N X & - N Y  \\ M X & M Y \ea + \ba{c} N \\ M \ea Q \ba{cc} \tilde M & \tilde N \ea
\end{align}
\end{subequations} where we used the formulae available in \cite[pp.110]{V}. Hence, 
we can use (\ref{e:930p050512})  to rewrite (\ref{e:928p050512}) as follows:

\hlinee
\begin{equation}
\label{e:931p050512}
\min_{\begin{matrix} (MX + M Q \tilde M ) \in \mathcal{S} \\ Q \in \mathbb{A}^{p \times m} \end{matrix} } \| T_1+T_2QT_3 \|
\end{equation} 

\hlinee

\noindent where $T_1$, $T_2$ and $T_3$ are obtained from (\ref{affine}).

The analysis above implies that  the sparsity constrained disturbance attenuation problem (as introduced in \cite[(1)/pp. 276 ]{Rotkowitz:2006kz}), or the sparsity constrained mixed $\mathcal{H}^2$ sensitivity problem (from \cite[pp. 139]{Zhou}) can be solved as a model--matching problem via  the numerical technique in \cite[Theorem~29]{Rotkowitz:2006kz}.

\section{Block-decoupling and streamlined solutions}
\label{Sec:5}

In this Section, we consider that the input and the output vectors of $G$ are partitioned into blocks so that, under certain conditions, $G$ can be factored in a special form that simplifies 
both the solution of the exact model matching problem of Theorem~\ref{mm} and the parametrization in Corollary~\ref{coro2}. Henceforward, we consider the following notation:

\begin{tabular}{l l}
& \\
 $r_y $ & number of partitions of $y$  \\
 $\{y_{[i]}\}_{i=1}^{r_y}$ & partitions of the output vector\\
 $m_i$ & dimension of $y_{[i]}$ 
 \end{tabular}
 
 \begin{tabular}{l l}
& \\
 $r_u $ & number of partitions of $u$  \\
$\{u_{[i]}\}_{i=1}^{r_u}$ & partitions of the input vector\\
 $p_i$ & dimension of $u_{[i]}$ \\
 &
 \end{tabular}

The partitions are constructed in a way that the following holds:

\begin{equation}
\label{partition}
\begin{matrix}
y^T & = & \left[ y_{[1]}^T \cdots y_{[r_y]} ^T \right]^T, & & \sum_{i=1}^{r_y} m_i & = & m\\ 
& & \\
u^T & = & \left[ u_{[1]}^T \cdots u_{[r_u]} ^T \right]^T, & & \sum_{i=1}^{r_u} p_i & = & p
\end{matrix}
\end{equation} 

Similarly, we also consider the partitioning of $G$ and $K$ as:

\begin{equation}
\label{partitieG}
\begin{matrix}
G  & = & \left[
\begin{matrix}
G_{[11]} & \cdots & G_{[1r_u]} \\
\vdots & \ddots & \vdots \\
 G_{[r_y1]} & \cdots & G_{[r_yr_u]}
 \end{matrix}
 \right] \\ & & \\
 K  & = & \left[
\begin{matrix}
K_{[11]} & \cdots & K_{[1r_y]} \\
\vdots & \ddots & \vdots \\
 K_{[r_u1]} & \cdots & K_{[r_ur_y]}
 \end{matrix}
 \right]
 \end{matrix}
 \end{equation}
 
 \begin{assumption} 
 \label{assumption:448pm27april2012}
 Throughout this Section, we assume that $G$ and an associated partition of the input and output (\ref{partition}) are given.
 \end{assumption}
 
 \begin{rem} Given factorizations of $G$ and $K$ as $G=\tilde{M}^{-1}\tilde{N}=NM^{-1}$ and $K=\tilde{X}\tilde{Y}^{-1}=Y^{-1}X$, respectively, the partition in (\ref{partition}) will induce a unique block-partition structure on the factors $N$, $M$, $\tilde{N}$, $\tilde{M}$, $X$, $Y$, $\tilde{X}$ and $\tilde{Y}$ as well.
 \end{rem}

\subsection{Input/Output Decoupled Coprime Factorizations}

We start by defining input and output decoupled factorizations for $G$.

\begin{defn} \label{ODLCF} Let $\tilde{N}$ and $\tilde{M}$  be a factorization of $G$. The pair $(\tilde{N},\tilde{M})$ is called output decoupled if $\tilde M$ has the following block diagonal structure:
\begin{equation}
\label{eq:147pm29apr12}
\tilde{M} = \diag ( \{ \tilde{M}_{[ii]} \}_{i=1}^{r_y} )
\end{equation} where $\diag ( \{ \tilde M_{[ii]} \}_{i=1}^{r_y} )$ is defined as:
\begin{equation} 
 \diag ( \{ \tilde{M}_{[ii]} \}_{i=1}^{r_y} ) \overset{def}{=} \left[ \begin{matrix} \tilde M_{[11]} & 0 & \cdots & 0 \\ 0 & \tilde M_{[22]} & \cdots & 0 \\ \vdots & &\ddots &\vdots \\ 0 & 0 & \cdots & \tilde M_{[r_yr_y]} \end{matrix} \right]
\end{equation}
\end{defn}

\begin{defn} \label{IDRCF} Let $N$ and $M$  be a factorization of $G$. The pair $(N,M)$ is called input decoupled if $M$ has the following block diagonal structure:
\begin{equation}
\label{eq:148pm29apr12}
M = \diag ( \{ M_{[ii]} \}_{i=1}^{r_u} )
\end{equation}
\end{defn}

\begin{rem}
\label{rem618pm27apr12}  Notice that an output decoupled factorization can always be constructed by factoring each block row of $G$ separately as follows:
\begin{equation}
\left[ G_{[i1]} \cdots G_{[i r_u]} \right] = \tilde M^{-1}_{[ii]} \left[ \tilde N_{[i1]} \cdots \tilde N_{[i r_u]} \right] , \qquad i \in \overline{1,r_y}
\end{equation} An input decoupled factorization can also be constructed by factoring the block columns of $G$.
\end{rem}

\begin{defn} \label{IOdecDCF} A DCF $\big(M, N$, $\tilde M, \tilde N$, $X, Y$, $\tilde X, \tilde Y\big)$ of $G$ is called input/output decoupled if the pairs $(N,M)$ and $(\tilde N, \tilde M)$ are input and output decoupled, respectively.
\end{defn}

It is important to notice that the procedure outlined in Remark~\ref{rem618pm27apr12} \underline{does not} guarantee that the pairs $(N,M)$ and $(\tilde N, \tilde M)$
will be co-prime, much less doubly co-prime. In fact, $G$ may not admit an input/output decoupled DCF. Sufficient conditions and algorithms to obtain an input/output decoupled DCF for $G$ are provided in the Appendix I.

There are two substantial benefits of working with an input/output DCF for G: The first is that the constraint on $Q$ in Theorem~\ref{truemain} reduces to $Q \in \mathcal{S} \bigcap \mathbb{A}^{p \times m}$, which leads to
a parametrization of all stabilizing controllers that has a simpler characterization. The second advantage is that the exact model-matching problem of Theorem~\ref{mm} admits a unique solution that
can be easily computed (see Section~\ref{sec:ModelMatching-2}).

\subsection{Theorem~\ref{truemain} Revisited}


Here, we {\bf modify the definitions of Section~\ref{Sec:3}} so that they account for the {\bf assumed} {\bf input/output}  {\bf partition} in ({\ref{partition}). 
More specifically, {\bf sparsity constraints will be imposed on entire block sub-matrices of} $K$. The definitions in
Section~\ref{Sec:3} can be recovered from the ones below for the case when the block sub-matrices have dimension one, i.e., provided that $r_y =m$ and $r_u =p$.

\begin{defn}
Given $K$ in $\mathbb{R}(\lambda)^{p \times m}$, we define $ \wp(K)  \in \mathbb{B} ^{r_u \times r_y}$  as follows:
\begin{equation}
\wp(K)_{ij}\overset{def}{=}\left \{
\begin{array}{l} 0  \quad \mathrm{if } \qquad  K_{[ij]}=0_{p_i \times m_j}  \\
                     1 \quad \mathrm{otherwise} \end{array} \right.  \qquad   i,j \in \overline{1,r_u}  \times \overline{1,r_y}
\end{equation} where $0_{p_i \times m_j}$ is a matrix with $p_i$ rows and $m_j$ columns and whose entries are all zero.
\end{defn}

\begin{defn} Conversely, for any binary matrix $\displaystyle K^\mathrm{bin}$  in $\mathbb{B} ^{r_u \times r_y}$, we define the following linear subspace:
\begin{equation}
\mathbb{S}(K^\mathrm{bin}) \overset{def}{=} \Big{\{} K \in  \mathbb{R}(\la)^{p\times m} \big{|} \wp (K) \leq K^\mathrm{bin} \Big{\}}
\end{equation}
\end{defn}


\begin{defn} Given $K^{\mathrm{bin}}$ in $\mathbb{B}^{r_u \times r_y}$, the {\em sparsity constraint} $\mathcal{S}$ is defined as:
\begin{equation}
\mathcal{S} \overset{def}{=}\mathbb{S}(K^{\mathrm{bin}}),
\end{equation}
\end{defn}

Hence,  $\mathcal{S}$ is the subspace of all controllers $\displaystyle K$ in $\displaystyle  \mathbb{R}(\la)^{p \times m}$ for which $K_{[ij]}=0_{p_i \times m_j}$ whenever $\displaystyle K^\mathrm{bin}_{ij}=0$. 
In addition, we assume that $\mathbb{S}$ and $\wp$ act on $G$ and $G^{\mathrm{bin}}$ as well as on the factors of any DCF of $G$ in an analogous way. 

\begin{rem}As a consequence of the definitions above, the following holds for any input/output decoupled DCF of $G$:
\begin{equation}
\begin{matrix}
\wp(M) & = & I_{r_u \times r_u}, & & \wp(N) &\underset{(a)}{\leq} & G^{bin} \\
\wp(\tilde M) & = & I_{r_y \times r_y}, & & \wp(\tilde N)  & \underset{(b)}{\leq} & G^{bin} 
\end{matrix}
\end{equation} where (a)-(b) follow from (\ref{eq:147pm29apr12}) and the fact that $G=\tilde M ^{-1}N = NM^{-1}$.
\end{rem}


The following Corollary is an immediate consequence of Theorem~\ref{truemain} and the facts that $\wp(M)  = I_{r_u \times r_u}$, $\wp(\tilde M) =  I_{r_y \times r_y}$,
and $M^{-1}$ and $\tilde M^{-1}$ are well defined TFMs.

\begin{coro}
\label{coro920pm29apr12}
Let  $\big(M, N$, $\tilde M, \tilde N$, $X, Y$, $\tilde X, \tilde Y\big)$ be an input/output decoupled DCF of $G$ and $\mathcal{S}$ be a 
 QI sparsity constraint.  
 \begin{itemize}
 \item {\bf Sufficiency:}
If $Q$ in $\mathbb{A}^{p\times m}$ is such that at least one of the following inequalities holds\footnote{In fact, it also follows from the statement of the Theorem that either both inequalities hold or none.}:
\begin{subequations}
\label{adouaEq}
\begin{align}
\wp (\tilde X_Q)  \leq  K^{\mathrm{bin}} \\ 
 \wp(X_Q)  \leq  K^{\mathrm{bin}}
\end{align}
\end{subequations} then $K_Q$ is a stabilizing controller in $\mathcal{S}$. 

\item {\bf Necessity:} If there is a stabilizing controller in $\mathcal{S}$ then there exists some $Q$ in $\mathbb{A}^{p\times m}$
for which both inequalities in  (\ref{adouaEq}) hold and, in addition, the controller can be written as $K_Q$.
\end{itemize}
\end{coro}

The following Corollary is the main result of this section.

\begin{coro} \label{cor:839pm29apr12} Let $\mathcal{S}$ be a given QI sparsity constraint and $(M,N,\tilde{M},\tilde{N},X,Y,\tilde{X},\tilde{Y})$
an input/output decoupled DCF of $G$. Assume that there is a stabilizing controller in $\mathcal{S}$ and let $Q_0$ in $\mathbb{A}^{p \times m}$ 
be selected to satisfy (\ref{adouaEq}). Any stabilizing controller in $\mathcal{S}$ can be written as $K_{Q}$, where
$Q$ is obtained as:
\begin{equation} 
Q=Q_0+Q_{\delta}, \qquad Q_{\delta} \in \mathcal{S} \cap \mathbb{A}^{p\times m}
\end{equation} 
\end{coro}
\begin{proof} From Corollary~\ref{coro920pm29apr12} and Theorem~\ref{mm}, it follows that since $Q_0$ satisfies (\ref{adouaEq}) then it will
also satisfy (\ref{mmEq}). Hence, from Corollary~\ref{coro2}, any stabilizing controller in $\mathcal{S}$ can be written as $K_Q$, with
$Q=Q_0+Q_{\delta}$, where $Q_{\delta}$ satisfies (\ref{Eq:1112amapril20}). The proof follows from noticing that since $M \otimes \tilde M$
is block diagonal and its inverse is a well defined TFM, $Q$ satisfies (\ref{Eq:1112amapril20}) if and only if $ \mathrm{vec} (Q_{\delta}) \in \mathrm{Null} (\Phi) \cap \mathbb{A}^{p \times m}$ holds,
or equivalently $Q_{\delta}$ is in $\mathcal{S}  \cap \mathbb{A}^{p \times m}$.
\end{proof}


\subsection{Theorem~\ref{mm} revisited }
\label{sec:ModelMatching-2}

In this subsection, we show that if an input/output decoupled DCF of $G$ exists then we can use Corollary~\ref{coro920pm29apr12} to obtain a simplified version of Theorem~\ref{mm}. A precise statement of this result is given 
in Corollary~\ref{atreia}.


\begin{defn}
We define the binary matrix $K^{\mathrm{bin}}_\perp $ belonging to the set  $\mathbb{B}^{r_u \times r_y}$ as follows:
\begin{equation} \label{Kbinperp}
\big( { K^{\mathrm{bin}}_\perp  }\big)_{ij}\overset{def}{=}\left \{
\begin{array}{l} 1  \quad \mathrm{if } \; K^{\mathrm{bin}}_{ij}=0,  \\
                     0 \quad \mathrm{otherwise}\: . \end{array} \right.
\end{equation}
\end{defn}

\begin{defn} Given $\displaystyle K^\mathrm{{bin}}_\perp  \in \mathbb{B}^{r_u \times r_y}$  we introduce the linear subspace $\mathcal{S}_\perp$ of $\displaystyle  \mathbb{R}(\la)^{p\times m}$ as
\begin{equation}
\label{DefnSperp}
\mathcal{S}_\perp \overset{def}{=}\Big{\{} K \in  \mathbb{R}(\la)^{p\times m} \Big{|} \: \wp(K)\leq K^{\: \mathrm{bin}}_\perp \Big{\}}.
\end{equation}
\end{defn}

\begin{coro}
\label{atreia} 
Let $\big(M, N$, $\tilde M, \tilde N$, $X, Y$, $\tilde X, \tilde Y\big)$ be an input/output decoupled DCF of $G$. Given a QI sparsity constraint $\mathcal{S}$, $G$ is stabilizable by a controller in $\mathcal{S}$ if and only if $M^{-1} \tilde{X}_{\mathcal{S}_{\perp}}$ is in $\mathbb{A}^{p \times m}$, where $\tilde{X}_{\mathcal{S}_{\perp}}$ results from the additive factorization $\tilde X = \tilde{X}_{\mathcal{S}} + \tilde{X}_{\mathcal{S}_{\perp}}$ satisfying $\wp(\tilde{X}_{\mathcal{S}}) \leq K^{\mathrm{bin}}$ and $\wp(\tilde{X}_{\mathcal{S}_{\perp}}) \leq K^{\mathrm{bin}}_{\perp}$.
\end{coro}

\begin{proof} We start by restating the first equation of (\ref{adouaEq}) for any $Q$ as:
\begin{equation}
\label{e115042712}
\wp(\tilde X+MQ) =  \wp \big (  \tilde{X}_{\mathcal{S}} + M Q_{\mathcal{S}} + M ( M^{-1} \tilde{X}_{\mathcal{S}_{\perp}} +   Q_{\mathcal{S}_{\perp}} )   \big ) \leq  K^{\mathrm{bin}}
\end{equation} where $Q = Q_{\mathcal{S}} + Q_{\mathcal{S}_{\perp}} $ and $\wp(Q_{\mathcal{S}}) \leq K^{\mathrm{bin}}$ and $\wp(Q_{\mathcal{S}_{\perp}}) \leq K^{\mathrm{bin}}_{\perp}$.

We now recall that according to Corollary~\ref{coro920pm29apr12}, $G$ is stabilizable by a controller in $\mathcal{S}$ if and only if there is $Q$ in $\mathbb{A}^{p \times m}$ so that (\ref{e115042712}) holds.
However, given the fact that $M^{-1}$ is block diagonal, (\ref{e115042712}) holds for some $Q$ in $\mathbb{A}^{p \times m}$ if and only if $M^{-1} \tilde{X}_{\mathcal{S}_{\perp}} = -   Q_{\mathcal{S}_{\perp}} $
has a solution where $Q_{\mathcal{S}_{\perp}}$ is in $\mathbb{A}^{p \times m}$. Since $M^{-1} \tilde{X}_{\mathcal{S}_{\perp}} = -   Q_{\mathcal{S}_{\perp}} $ has a unique solution because $M$ is
invertible, we conclude that there exists $Q_{\mathcal{S}_{\perp}}$ is in $\mathbb{A}^{p \times m}$ satisfying (\ref{e115042712}) if and only if $M^{-1} \tilde{X}_{\mathcal{S}_{\perp}}$ is in $\mathbb{A}^{p \times m}$.
Notice that  $\wp(M^{-1} \tilde{X}_{\mathcal{S}_{\perp}}) \leq K^{\mathrm{bin}}_{\perp}$ always holds because $M$ is block diagonal.
\end{proof}

 \section{Conclusions} 
 \label{Sec:6}
We address the design of stabilizing controllers subject to a pre-selected quadratically invariant sparsity pattern. We show that the previously unsolved problem of determining stabilizability with sparsity constraints is equivalent to the solvability of an exact model-matching system of equations that is tractable via existing techniques, and we also outline a systematic method to compute an admissible controller. The proposed analysis also leads to a convex parametrization that is an extension of Youla's classical result so as to incorporate sparsity constraints on the set of stabilizing controllers. We indicate how this parametrization can be used to write sparsity-constrained norm-optimal control problems in convex form.

\section*{Appendix I}

This Appendix has two parts. In the first part  we give a sufficient condition that guarantees the existence of an output (input) decoupled left (right) {\em coprime} factorization for $G$, as  in Definitions~\ref{ODLCF}  and~\ref{IDRCF}. In the second part, we show that if  $G$ admits the aforementioned factorizations then there is a state--space method for computing its input/output decoupled DCF of Definition~\ref{IOdecDCF}. 

\subsection{A Sufficient Condition for the Existence of an Output (Input) Decoupled Left (Right) Coprime Factorization}
 We are given a plant $G$, partitioned as in (\ref{partitieG}).  As described in Remark~\ref{rem618pm27apr12},  we perform a \underline{left coprime} factorization for each of the $r_y$ block--rows of $G$ (such left coprime factorizations always exist and can be computed using the classical state--space methods  from \cite{NETT:1984un, Khargonekar:1982vm}) in order to get
\begin{equation} \label{row}
\left[ G_{[i1]} \cdots G_{[i r_u]} \right] = \tilde M^{-1}_{[ii]} \left[ \tilde N_{[i1]} \cdots \tilde N_{[i r_u]} \right] , \qquad i \in \overline{1,r_y}.
\end{equation}
\noindent  where the poles of $M_{[ii]}$ can be placed anywhere in the stability domain $\Omega$. 

The following proposition gives a necessary and sufficient condition under which the row factorizations in (\ref{row}) can be concatenated to produce a left decoupled coprime factorization for $G$. It should be noted that a left decoupled coprime factorization for $G$ may exist that cannot be constructed from the row factorizations in (\ref{row}). This indicates that the proposition is only a sufficient condition for the existence of a left decoupled coprime factorization for $G$.


\begin{prop} Let $(\tilde{M},\tilde{N})$ be an output decoupled left factorization derived from the row coprime factorizations (\ref{row}) as follows:

\begin{equation}
 \left[
\begin{matrix} \label{ec1234}
G_{[11]} & \cdots & G_{[1r_u]} \\
\vdots & \ddots & \vdots \\
 G_{[r_y1]} & \cdots & G_{[r_yr_u]}
 \end{matrix}
 \right] = \left[ \begin{matrix} \tilde M^{-1}_{[11]}  & \cdots & 0  \\ \vdots & \ddots &\vdots \\ 0  & \cdots & \tilde M^{-1}_{[r_yr_y]} \\ \end{matrix}
 \right]   \left[  \begin{matrix} \tilde N_{[11]} & \cdots & \tilde N_{[1 r_u]} \\ \vdots & \ddots & \vdots \\  \tilde N_{[r_y1]} & \cdots & \tilde N_{[r_y r_u]} \\ \end{matrix} \right] 
\end{equation} and consider $\Psi$ to  be the following TFM:
\[
\Psi \overset{def}{=}\left[ \begin{matrix} \tilde M_{[11]}  & \cdots & 0 &  \tilde N_{[11]} & \cdots & \tilde N_{[1 r_u]}  \\ \vdots  & \ddots  &\vdots  & \vdots & & \vdots \\ 0  & \cdots & \tilde M_{[r_yr_y]} & \tilde N_{[r_y1]} & \cdots & \tilde N_{[r_y r_u]}   \end{matrix} \right] 
\]
The following holds:
\begin{enumerate}
\item The output decoupled left factorization in (\ref{ec1234}) is coprime if and only if the following holds:
\begin{equation}
\label{testolc}
\rank(\Psi(\lambda)) = m, \qquad \lambda \in \Lambda_G 
\end{equation} where $\Lambda_G$ represents the set of unstable poles of $G$.
\item The condition in (\ref{testolc}) does not depend on the choice of the row coprime factorizations (\ref{row}).
\end{enumerate}
\end{prop}

\begin{proof} The proof follows as a consequence of standard results in linear systems theory, so we only provide a sketch of the ideas. We start by invoking a result used in \cite{Oara:1999vz}  that the left factorization $(\tilde M, \tilde N)$  is coprime if and only if $\Psi$ has no Smith zeros\footnote{  A complex number $\lambda_0 \in \mathbb{C}$ is called a Smith zero of $\Psi$ if $\Psi(\lambda_0)$ is not full--row rank. } in $\mathbb{C} - \Omega$ . Hence, from the statement in 1), we are left to prove that any Smith zero of $\Psi$ in $\mathbb{C} - \Omega$ is a pole of $G$. The proof of 1) is concluded by noticing that if a given $\lambda_0$  in $\mathbb{C} - \Omega$ is not a pole of $G$ then, from the coprimeness of the row factorizations in (\ref{row}), $\tilde{M}(\lambda_0)$ is invertible and hence full rank, leading to the conclusion that $\Psi(\lambda_0)$ must also be full row rank, and hence $\lambda_0$ is not a Smith zero of $\Psi$.

It only remains to prove 2). The argument here follows from the fact that the set of {\em all} left coprime factorizations of any block--row of $G$ is given by (\ref{row}) up to a premultiplication by a unimodular TFM \cite[Ch.~4, Theorem~43]{V}. This in turn implies that $\Psi(\la)$ is unique, up to a premultiplication of a block--diagonal, unimodular TFM (having the same block partition as $\tilde M$), which does not alter the rank condition on $\Psi(\la)$ at any unstable $\la_0 \in (\mathbb{C}-\Omega)$.
\end{proof}

The corresponding test for the existence of input decoupled right coprime factorization of $G$ (Definition~\ref{IDRCF}) is analogous and is therefore omitted for brevity. 


\subsection{A State-Space Method to Compute an Input/Output Decoupled DCF}
We assume that $G$ admits an output decoupled left coprime factorization $G=\tilde M ^{-1} \tilde N$ and an input decoupled right coprime factorization $G= N M^{-1}$.  These factorizations must be pre-computed (using for instance the arguments described in Appendix I- A), so we consider the $\tilde M, M, \tilde N, N$ factors  fixed. Under these conditions, here we provide a computational algorithm to obtain the Input/Output Decoupled DCF of $G$  (Definition~\ref{IOdecDCF}), {\em i.e.} a DCF  

\bb{dcrel2}
\ba{cc}  Y &  X \\ -\tilde N &  \tilde M\ea
\ba{cc} M & - \tilde X \\ N &  \tilde Y \ea = I_{m+p}.
\ee

\noindent containing the fixed factors $\tilde M, M, \tilde N, N$.  The fact that the coprime factorization (\ref{dcrel2})  always exists is guaranteed by \cite[Ch.~4, Theorem~60]{V} but since we are not aware of any method to actually compute it, we will present one here.

 The following additional notation is needed: given any $n$--dimensional state--space representation $(A$, $B$, $C$, $D)$ of  a LTI system, its input--output description  is given by the {\em transfer function matrix} (TFM) which is the $m\times p$ matrix with real, rational functions entries

\begin{equation}
\label{state-space}
G =  \ba{c|c}A & B \\ \hline C & D \ea \overset{def}{=} D + C(\la I_n - A)^{-1}B,
\end{equation}

\noindent where  $A, B, C, D$ are $n\times n$, $n \times p$, $m \times n$, $m \times p$
real matrices, respectively  while  $n$ is also called {\em the order} of the realization.  For elementary notions in linear systems theory,
such as controllability, observability, detectability, we refer to \cite{Wonham}, or any other standard text book in linear systems.



We have started out with  an output decoupled left  coprime factorization (Definition~\ref{ODLCF})   of the plant $G=\tilde{M}^{-1}   \tilde{N}$.
The state--space representation of this  left  coprime factorization can be obtained according to Proposition~\ref{Oara:1999vz99}~{\bf A)} from Appendix II,  starting from a {\em certain}  stabilizable state--space realization of $G$  (which we take without loss of generality to be in the Kalman Structural Decomposition, \cite{Kalman:1962uv}) and which we consider fixed:

\begin{equation}
\label{Gdetect}
G=\ba{cccc|c} \star & \star & \star   & \star &  \star  \\  
                                 O & A_{22} & O   & A_{24} &  B_2  \\
                                 O & O          & \star & \star &  O   \\
                                  O & O &  O & A_{44} &  O  \\ 
                    \hline   O & C_2 &  O & C_4 &  D  \\ \ea
\end{equation}

\noindent  with the $\star$ denoting parts of the realizations that are of no importance in this proof. Continuing with Proposition~\ref{Oara:1999vz99}~{\bf A)} from Appendix~II, there also  exists an invertible matrix $U$ and a feedback matrix $F$  (both fixed)   such that  (with $F$ partitioned in accordance with (\ref{Gdetect})) we get

\begin{equation}
\label{LCFSS}
\ba{cc} - \tilde N &  \tilde M \ea=U{}^{-1}\ba{cccc|cc} 
                                 \star & \star & \star   & \star &  \star  & \star  \\  
                                 O & A_{22} -F_2C_2 & O   & A_{24}-F_2 C_4 &  B_2  - F_2 D & F_2 \\
                                 O & \star   & \star & \star &  \star & \star  \\
                                  O & -F_4 C_2 &  O & A_{44} -F_4 C_4 &   -F_4 D & F_4 \\ 
                    \hline   O &  -C_2 &  O & -C_4 &  -D & I \\ \ea
\end{equation}

\noindent with

\begin{equation} 
\label{Spectrul1}
\Lambda \left(  \ba {cc} A_{22} -F_2 C_2   & A_{24}-F_2 C_4  \\
                              -F_4 C_2 & A_{44} - F_4 C_4  \\   \ea \right) \subset \Omega.
\end{equation}

\noindent  Note that since (\ref{Gdetect}) is stabilizable it follows that $\Lambda(A_{44})\subset \Omega$. After removing the unobservable part from  (\ref{LCFSS}) we get that
\begin{equation}
\label{LCFSSmin}
\ba{cc} - \tilde N  &  \tilde M \ea=U {}^{-1}\ba{cc|cc} 
                          A_{22} - F_2 C_2   & A_{24}-F_2 C_4 &  B_2  - F_2 D & F_2 \\
                              -F_4 C_2 & A_{44} - F_4 C_4 &   -F_4 D & F_4 \\ 
                    \hline     -C_2  & -C_4 &  -D & I \\ \ea
\end{equation}

We have also started out with  an input decoupled right coprime factorization (Definition~\ref{IDRCF})  $G=NM^{-1}$. According to Proposition~\ref{Oara:1999vz99}~{\bf B)} in Appendix~II, there exists a {\em certain} detectable state--space realization of $G$ (which we take without loss of generality to be in the Kalman Structural Decomposition) and which we also consider fixed: 

\begin{equation}
\label{Gstabiliza}
G=\ba{cccc|c} A_{11}  & A_{12} & \star   & \star &  B_1  \\  
                                 O & A_{22} & O   & \star &  B_2  \\
                                 O & O          & \star & \star &  O   \\
                                  O & O &  O & \star &  O  \\ 
                    \hline   O & C_2 &  O & \star &  D  \\ \ea
\end{equation}

\noindent with the $\star$ denoting parts of the realization that are of no importance here.  Any two realizations of $G$  will always have the same  the controllable and observable part,  up to a similarity transformation - that is to say that if the controllable and stabilizable part of  (\ref{Gdetect}) is  $(A_{22}, B_2, C_2,D)$ then the controllable and stabilizable part of  (\ref{Gstabiliza}) must be  $(Z^{-1}A_{22}Z, Z^{-1} B_2,$ $C_2Z,D)$, for some invertible, real matrix $Z$ . We can apply this similarity transformation adequately on (\ref{Gstabiliza}), such that the the controllable and stabilizable part   $(A_{22}, B_2, C_2,D)$, appears identical on both realizations (\ref{Gdetect}) and (\ref{Gstabiliza}), respectively. This simplifies future computations.

According to  Proposition~\ref{Oara:1999vz99}~{\bf B)} from Appendix~II, along with realization (\ref{Gstabiliza}), there also  exists an invertible matrix $V$ and a feedback matrix $L$  (both fixed)   such that  (with $L$ partitioned in accordance with (\ref{Gstabiliza}))

\begin{equation}
\label{RCFSS}
\ba{c} M \\ N \\   \ea=\ba{cccc|c} 
                                 A_{11}- B_1L_1 & A_{12} -B_1L_2 & \star   & \star &  B_1  \\  
                             -  B_2 L_1 & A_{22} - B_2L_2 & \star   & \star &  B_2 \\    
                                 O & O   & \star & \star  & O  \\
                                  O & O   & O & \star  & O  \\
                    \hline   -L_1 &  -L_2 &  \star & \star &   I \\ 
                                 - DL_1 & C_2- D L_2 & \star & \star &D \\  \ea V
\end{equation}

\noindent with

\begin{equation} 
\label{Spectrul2}
\Lambda \left(    \ba{cc}         A_{11}- B_1L_1 & A_{12} - B_1L_2  \\  
                               - B_2 L_1 & A_{22}-B_2 L_2 \ea \right) \subset \Omega, 
\end{equation}
\noindent Note that since (\ref{Gstabiliza}) is detectable it follows that $\Lambda(A_{11})\subset \Omega$.  After removing the uncontrollable part from  (\ref{RCFSS}) we get that

\begin{equation}
\label{RCFSSmin}
\ba{c} M \\ N \\   \ea=\ba{cc|c} 
                                 A_{11} - B_1L_1 & A_{12} - B_1L_2 &  B_1  \\  
                               - B_2 L_1 & A_{22}-B_2L_2 &  B_2 \\    
                    \hline   -L_1 & - L_2 &    I \\ 
                                 - DL_1 & C_2- D L_2 &  D \\  \ea V
\end{equation}

We have come now  to the following stabilizable and detectable state--space realization of $G$, which we consider fixed:

\begin{equation}
\label{GMare}
G=\ba{ccc|c} A_{11}  & A_{12}    & \star &  B_1  \\  
                                 O & A_{22}   & A_{24} &  B_2  \\
                                  O & O  & A_{44} &  O  \\ 
                    \hline   O & C_2 & C_4 &  D  \\ \ea
\end{equation}

Since $\Lambda(A_{11})\subset\Omega$ we get that (\ref{GMare}) is detectable and since $\Lambda(A_{44})\subset\Omega$ we get that (\ref{GMare}) is stabilizable, hence (\ref{GMare}) satisfies the hypothesis  from Theorem~\ref{LucicThm} {\bf iii)} from Appendix~II. Starting from realization (\ref{GMare}) (which is fixed),   (\ref{dc1final}) and (\ref{dc2final}) yield a valid DCF of $G$ for  any stabilizing feedback matrices $F^+$ and $L^+$ (partitioned in accordance with (\ref{GMare}) and satisfying Theorem~\ref{LucicThm} {\bf ii)} from Appendix~II), and any invertible matrix $T^+$ satisfying the hypothesis of Theorem~\ref{LucicThm} {\bf i)}.  We will denote the factors of this particular DCF with $\big(M^+, N^+$, $\tilde M^+, \tilde N^+$, $X^+, Y^+$, $\tilde X^+, \tilde Y^+\big)$. 
After removing the unobservable part, the $\tilde M^+$ factor will be (the computation are similar with those for getting from (\ref{LCFSS}) to (\ref{LCFSSmin}))

\begin{equation}
\label{Dtilda}
  \tilde M^+  ={U^+}^{-1}\ba{cc|c} 
                          A_{22} - F^+_2 C_2   & A_{24}-F^+_2 C_4 &   F^+_2 \\
                              -F^+_4 C_2 & A_{44} - F^+_4 C_4 &   F^+_4\\ 
                    \hline     -C_2  & -C_4 &  I \\ \ea 
\end{equation}

\noindent where

\begin{equation} 
\label{Spectrul3}
\Lambda \left(  \ba {cc} A_{22} -F^+_2 C_2   & A_{24}-F^+_2 C_4  \\
                              -F^+_4 C_2 & A_{44} - F^+_4 C_4  \\   \ea \right) \subset \Omega.
\end{equation}

\noindent and $U^+$ is a real, invertible matrix. We compute the factor $\tilde \Theta \overset{def}{=} \tilde M  {\tilde {M^+}} ^{-1}$ using the state--space realizations from (\ref{LCFSS}) and (\ref{Dtilda}) respectively and we get

\begin{equation}
\label{nasoala}
 \tilde \Theta  =U^{-1}\ba{cccc|c} 
                          A_{22} -F_2C_2   & A_{24}-F_2 C_4 &   F_2C_2 & F_2C_4 & F_2 \\
                             - F_4 C_2 & A_{44} - F_4 C_4 &   F_4C_2 & F_4 C_4 & F_4 \\ 
                              O & O & A_{22} & A_{24}  &F^+_2 \\
                              O & O & O & A_{44}  & F^+_{4} \\
                    \hline     -C_2  & -C_4 &   C_2 & C_4 & I \\ \ea U^+.
\end{equation}

  After removing the unobservable part from (\ref{nasoala}) we get that

\begin{equation}
\label{Deltatilda}
  \tilde \Theta  =U^{-1}\ba{cc|c} 
                          A_{22} -F_2C_2   & A_{24}-F_2 C_4 &   F_2 -F^+_2 \\
                              -F_4 C_2 & A_{44} - F_4 C_4 &   F_4 - F^+_4\\ 
                              \hline     -C_2  & -C_4 &  I \\ \ea U^+
\end{equation}

\noindent and consequently

\begin{equation}
\label{Deltatildainversa}
  \tilde \Theta^{-1}  ={U^+}^{-1}\ba{cc|c} 
                          A_{22} -F^+_2C_2   & A_{24}-F^+_2 C_4 &   F_2 -F^+_2 \\
                              -F^+_4 C_2 & A_{44} - F^+_4 C_4 &   F_4- F^+_4\\ 
                              \hline     C_2  & C_4 &  I \\ \ea U,
\end{equation}

\noindent which combined with (\ref{Spectrul1}) and (\ref{Spectrul3}) shows that $\tilde \Theta$ is unimodular.  A similar line of reasoning can be used to prove that $\Theta\overset{def}{=}{M^+}^{-1}M$ is unimodular.

 Finally, compute
\begin{equation}
\label{ArbDenEq2}
\Big(\ba{cc}\Theta^{-1} & O\\O & \tilde{\Theta}\ea\ba{cc} {\tilde Y^+} & {\tilde X^+} \\ -{\tilde N^+} &  {\tilde M^+}\ea\Big)
\Big(\ba{cc} M^+ & - X^+ \\ N^+ &  Y^+ \ea \ba{cc}\Theta & O\\O & \tilde{\Theta}^{-1}\ea\Big)= I_{m+p}
\end{equation}

\noindent which is still a  DCF of $G$  in its own right, due to the unimodularity of $\Theta$ and $\tilde \Theta$. Plugging in the definitions of $\tilde \Theta$ and $\Theta$ into (\ref{ArbDenEq2}) yields

\begin{equation}
\label{ArbDenEq3}
\ba{cc} \Theta^{-1} {\tilde Y^+} &  \Theta^{-1} {\tilde X^+} \\ -{\tilde N} &  {\tilde M} \ea
\ba{cc} M & - X^+ {\tilde \Theta}^{-1} \\ N & Y^+   {\tilde \Theta}^{-1}  \ea= I_{m+p}
\end{equation}

\noindent which is an input/output decoupled DCF of $G$ and the algorithm ends.

\section*{Appendix II}

\begin{theorem}
\cite[Theorem 1]{Lucic:2001gk}
\label{LucicThm}
Let $G$ be some proper $m \times p$ TFM. The class of {\em all} DCFs (\ref{dcrel}) of $G$ over $\Omega$ is given by

\label{dcfinal}
\bb{dc1final}
\ba{cc} M &  -\tilde X \\ N &   \tilde Y \ea = \ba{c|cc} A - BL  & B  & F \\ \hline -L & I & 0 \\ C-DL & D & I \ea
T \: ,\quad
\ee
\bb{dc2final}
\ba{cc}  Y &  X \\ -\tilde N &  \tilde M \ea
 = T^{-1}\ba{c|cc} A - FC & B-FD & F \\ \hline L & I & 0 \\
-C & -D & I \ea,
\ee

\noindent where $A, B, C, D, F, L$ and $T$ are real matrices accordingly dimensioned such that

\noindent {\bf \em i)} $T = \ba{cc}V & W \\ O  & U \ea$  has its diagonal $p \times p$ block $V$ and $m \times m$ block $U$ respectively, invertible,

\noindent {\bf \em ii)} $F$ and $L$ are feedback--matrices such that $\displaystyle \La(A  -BL) \cup \La(A -FC) \subset \Omega$,

\noindent {\bf \em iii) }
$G =  \ba{c|c}A & B \\ \hline C & D \ea$ is a stabilizable and detectable realization.
\end{theorem}


\begin{coro}
\label{Oara:1999vz99}
Let $G$ be an arbitrary $m\times p$ TFM and $\Omega$ a domain in $\boC$.

\noindent {\bf \em A)} The class of {\em all} left coprime factorizations of $G$ over $\Omega$, $G = \tilde{M}^{-1}\tilde{N}$, is given by
\bb{dc2}
\ba{cc} \tilde N &  \tilde M \ea
 = U^{-1}\ba{c|cc} A - FC & B-FD & -F \\ \hline
C & D & I \ea,
\ee

\noindent where $A, B, C, D, F$ and $U$ are real matrices accordingly dimensioned such that

\noindent {\bf \em i)}  $U$ is any $m \times m$ invertible matrix,

\noindent {\bf \em ii)}  F is any feedback matrix that  allocates the observable modes of the $(C,A)$ pair to $\Omega$,

\noindent {\bf \em iii)} $G =  \ba{c|c}A & B \\ \hline C & D \ea$ is a stabilizable realization.

\noindent {\bf \em B)} The class of {\em all} right coprime factorizations of $G$ over $\Omega$, $G=NM^{-1}$ is given by
\bb{dc1}
\ba{c} M \\ N \ea = \ba{c|c} A - BL  & B   \\ \hline -L & I  \\ C-DL & D \ea V
\ee

\noindent where $A, B, C, D, L$ and $V$ are real matrices accordingly dimensioned such that

\noindent {\bf \em i)}  $V$ is any $p \times p$ invertible matrix,

\noindent {\bf \em ii)} L is any feedback matrix that  allocates the controllable modes of the $(A,B)$ pair to $\Omega$,

\noindent {\bf \em iii)} $G =  \ba{c|c}A & B \\ \hline C & D \ea$ is a detectable realization.
\end{coro}


The proof of Corollary~\ref{Oara:1999vz99} follows on the lines of Theorem~\ref{LucicThm}.

\bibliographystyle{IEEEtran}

\bibliography{AutomatedRefs,AutomatedRefs2,ManualEntries}

\end{document}